\documentclass[conference]{IEEEtran}
\usepackage{etoolbox}
\usepackage{amsmath}
\usepackage{tikz}
\usepackage{lipsum}
\usepackage{graphicx,epstopdf,tikz,pgfplots,subfig,amssymb,amsfonts,amsmath,mathrsfs,lettrine}
\usepackage[makeroom]{cancel}
\usepackage[T1]{fontenc}
\usepackage[utf8]{inputenc}
\usepackage{authblk}
\usepackage[font=small,labelfont=bf]{caption}
\usepackage{float}
\usepackage{cite}
\usepackage{url}

\ifCLASSINFOpdf

\else

\fi
\newtheorem{theorem}{Theorem}
\newtheorem{lem}{Lemma}

\newtheorem{col}{Corollary}
\newcommand{\ignore}[1]{}

\begin{document}
%
\title{Cooperative Handover Management in Dense Cellular Networks}

\author[Rabe Arshad]
       {Rabe Arshad$^{\star}$, Hesham ElSawy$^*$, Sameh Sorour$^{\dagger}$,
       Tareq Y. Al-Naffouri$^*$, and Mohamed-Slim Alouini$^*$
       \\ \\
       $^{\star}$Electrical Engineering Department, King Fahd University of Petroleum and Minerals (KFUPM), Saudi Arabia\\
       Email: g201408420@kfupm.edu.sa\\
       $^*$King Abdullah University of Science and Technology (KAUST), Saudi Arabia\\
       Emails: $\{$hesham.elsawy, tareq.alnaffouri, slim.alouini$\}$@kaust.edu.sa\\
$^{\dagger}$Department of Electrical and Computer Engineering, University of Idaho, USA\\
       Email: samehsorour@gmail.com}

\maketitle

\vspace{-1cm}
\begin{abstract}
Network densification has always been an important factor to cope with the ever increasing capacity demand. Deploying more base stations (BSs) improves the spatial frequency utilization, which increases the network capacity. However, such improvement comes at the expense of shrinking the BSs' footprints, which increases the handover (HO) rate and may diminish the foreseen capacity gains. In this paper, we propose a cooperative HO management scheme to mitigate the HO effect on throughput gains achieved via cellular network densification. The proposed HO scheme relies on skipping HO to the nearest BS at some instances along the user's trajectory while enabling cooperative BS service during HO execution at other instances. To this end, we develop a mathematical model, via stochastic geometry, to quantify the performance of the proposed HO scheme in terms of coverage probability and user throughput. The results show that the proposed cooperative HO scheme outperforms the always best connected based association at high mobility. Also, the value of BS cooperation along with handover skipping is quantified with respect to the HO skipping only that has recently appeared in the literature. Particularly, the proposed cooperative HO scheme shows throughput gains of 12\% to 27\% and 17\% on average, when compared to the always best connected and HO skipping only schemes at user velocity ranging from 80 km/h to 160 Km/h, respectively.
\end{abstract}

\IEEEpeerreviewmaketitle

\section{Introduction}

\lettrine{N}{       etwork} densification is a potential solution to cater the increasing traffic demand and is expected to have a major contribution in fulfilling the ambitious 1000-fold capacity improvements required for next generation 5G cellular networks~\cite{1}. Network densification improves the spatial frequency reuse by shrinking the BSs' footprints to increase the delivered data rate per unit area. Hence, each BS serves lesser number of users with higher throughput for each user. However, such improvement comes at the expense of increased handover (HO) rates for mobile users. Mobile users change their BS associations more frequently in denser network environment due to the reduced BSs' footprints, to maintain the best connectivity. Note that the best network connectivity may differ according to the network objective~\cite{4a,5a}. However, in all cases, densifying the network by deploying more BSs, decreases the service region of each BS and increases the handover rate.

The HO procedure involves signaling between mobile user, serving BS, target BS and the core network, which consumes physical resources and incurs delay. Therefore, the number of HOs per user per unit time is always a performance limiting parameter for cellular operators. The effect of HO is more acute in highly dense cellular networks due to the excessive handover rate, which may negate the expected gains from network densification. In extreme cases, where high mobility exists in urban regions, such as users riding monorails in downtowns, the cellular networks may fail to support users due to the small cell dwell times. Several studies were conducted about HO management in dense cellular networks including \cite{HO_1}, \cite{HO_2}. A recent study~\cite{icc} proposes a new HO scheme, denoted as HO skipping, to reduce HO signaling and enhance the average rate for mobile users. The main idea in the HO skipping scheme is to sacrifice the best connectivity at some points in the user's trajectory to reduce the number of HOs per unit time. That is, a mobile user may skip a handover to the nearest BS and connect to a farther BS along its trajectory in order to maintain a longer connection and experience a better long term average throughput. Despite the lower service rates, the users obtain at times, when they are not connected to their nearest BSs (which we refer to as the blackout phases), the overall throughput gains of HO skipping were indeed observed in [6] due to HO delay reduction. Yet, the matter of improving the service rates (and thus the overall throughput) of users during the blackout phases is still to be addressed to guarantee a ubiquitous acceptable quality of service.

In this paper, we propose a cooperative handover management scheme that exploits both HO skipping to reduce the HO rate and BS cooperation to enhance the performance during blackout phases. The BSs cooperate by forming a network MIMO system via non-coherent coordinated multipoint (CoMP) transmission strategy \cite{3gpp, crancomp1, crancomp2}. When the user decides to skip the handover to the nearest BS, the serving BS and the target BS simultaneously transmit the user data during its transition through the skipped BS coverage area as shown in Fig.~\ref{1d}. Particularly, in blackout phase, the second and third nearest BSs cooperate to serve the user in order to improve the signal-to-interference-plus-noise-ratio (SINR). Consequently, the cooperative handover management scheme simultaneously reduces HO delay and maintains high SINR for the user along its trajectory. It is worth mentioning that the non-coherent transmission is considered in this paper as it may be hard to estimate the channel state information (CSI) in the considered high mobility scenario.

In order to draw rigorous conclusions on the proposed handover management scheme, we develop a mathematical model, based on stochastic geometry, which incorporates the HO effect on the coverage probability and throughput. Stochastic geometry is a powerful mathematical tool for modeling, designing and analyzing the performance of wireless networks encompassing random topologies (see \cite{6a} for a survey). We focus on the performance of a cooperative handover management scheme in a single tier downlink cellular network with BSs modeled via a Poisson point process (PPP). The PPP assumption is widely accepted for modeling cellular networks and has been verified in \cite{7a,8a,9a} by several experimental studies. To this end, we derive expressions for the coverage probability and the average throughput in the proposed HO skipping scheme and compare it to the best connected scheme as well as the non-cooperative HO skipping scheme proposed in \cite{icc}. The results show considerable improvement in coverage probability for the HO skipping case when BS cooperation is enabled. Also, the gains in average throughput are achieved at lower user velocity when compared to the HO skipping without CoMP transmission. Finally, we quantify the performance loss due to non-coherent COMP transmission when compared to the coherent CoMP transmission (i.e. with precoding), which is shown to be less than $6\%$ at low SINR thresholds.

\begin{figure}
\centering
\includegraphics[width=0.9\linewidth]{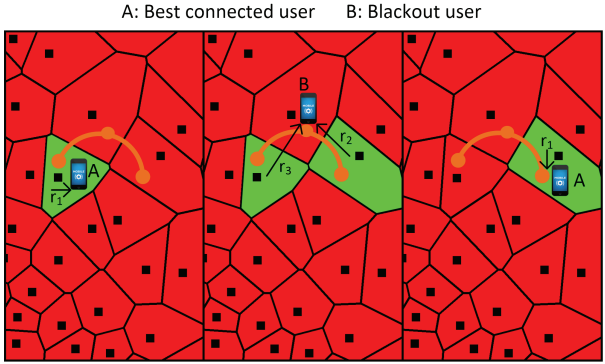}
\small \caption{A and B represent best connected and blackout with BS cooperation users, respectively. $r_i$ represents the distance between the user and $i^{th}$ nearest BS. Red and green colors show interference and serving regions with BSs located at the corresponding voronoi centers, respectively. Brown solid line represents the user trajectory.}
\label{1d}
\end{figure}
\section{System Model}
In this paper, we consider a single tier downlink cellular network with CoMP transmission. We abstract BSs' locations according to a homogeneous PPP $\Phi$ with BS intensity $\lambda$. All BSs have the same transmit power denoted by $P$. A general path loss model with path loss exponent $\eta>2$ is considered. In addition to the path loss, the channel introduces multi-path fading in the transmitted signal. Channel gains are assumed to have Rayleigh distribution with unit power i.e. $h\sim exp(1)$.
For the SINR analysis to be tractable, we infer the average SINR along the users trajectory via the spatially averaged stationary SINR in the network. Without loss of generality, the stationary SINR is obtained by studying the SINR of a test user located at the origin~\cite{book}. For the handover analysis, we assume that the test user is moving on an arbitrary long trajectory with velocity $v$. We assume that a HO is performed when a user enters the voronoi boundary of a particular BS. Also, we ignore all types of HO failures and assume that all HOs are successful. During each HO, a delay $d$ is incurred to execute HO signaling between the user, the serving BS, the target BS and the core network. We consider different backhauling schemes that impose different HO and core network signaling delays.

Fig. 1 shows the conventional and proposed handover schemes. In the conventional case, which we refer to as the best connected scenario, the user is always connected to the BS that provides the strongest received signal power. Hence, a handover is executed at every cell boundary crossing and the user is always associated to its nearest BS. In the proposed scheme, the user skips associating to the nearest BS at the first cell boundary crossing. The second and third nearest BSs thus cooperate to serve this user via non-coherent CoMP transmission. At the second cell boundary crossing, the user keeps its association to the nearest BS only. This pattern is repeated along the user's trajectory. Hence, on average, the user spends 50\% of the time in the best connected case and 50\% of the time in blackout case with BS cooperation. Let $\{r_i; i=1,2,3\}$ be the set of ascending ordered distances between the user and all BSs. The user is thus served by the BS located at the distance $r_1$ $50 \%$ of the time and is simultaneously served by the BSs located at $r_{2}$ and $ r_{3}$ for rest of the time. We present our model assuming that the BSs do not have CSI as it is difficult to estimate the channel for mobility scenarios. For the sake of completeness, the non-coherent CoMP transmission is benchmarked against the coherent CoMP transmission where perfect CSI is assumed to be available.
\section{Distance Distribution}
In stochastic geometry analysis, the first step is to characterize the service distance distribution, which is further used to characterize the coverage probability. In the proposed scheme, the user is either served from the nearest BS, or the second and third nearest BSs. Hence, we first derive the joint distribution of the distances between the user to the nearest, second nearest, and third nearest BS(s), which is illustrated in the following lemma:
\begin{lem}
\label{lem:joint}
In a single tier cellular network with intensity $\lambda$, the joint distance distribution between a user and its skipped and serving BSs with cooperation is given by
\small
\begin{align}
\hspace{-0.3cm}f^{(bk)}_{r_1,r_2,r_3}(x,y,z)\hspace{-0.05cm}=\hspace{-0.05cm}{(2\pi \lambda)^{3}  x y z e^{-\pi \lambda z^2}}; 0 \leq x \leq y \leq z \leq \infty
\label{eq:joint}
\end{align}
\normalsize
\end{lem}
\begin{proof}
\textit{By conditioning on $r_3$, the joint conditional distribution of $r_1$ and $r_2$ is the order statistics of two $iid$ random variables with PDF $\frac{2r}{r_{3}^{2}}$, where $0 \leq r \leq r_3$. The joint conditional distribution is given by $f_{r_1,r_2}(x,y \vert r_3) = \frac{8xy}{r_{3}^{4}}$, where $0<x<y<r_3$. By following Bayes' theorem, the joint PDF $f^{(bk)}_{r_1,r_2,r_3}(.,.,.)$ is obtained by multiplying the conditional joint PDF of $r_1$ and $r_2$ by the marginal PDF of $r_3$. The lemma follows by performing this marginalization over $r_3$, using its marginal distribution derived in eq. (2) in~\cite{dist}.}
\end{proof}
The marginal distance distributions for the blackout case with and without BS cooperation are characterized by the following corollary:
\begin{col} \label{col_dist}
The marginal PDF of the distance between a test user and its serving BS in the best connected case is given by:
\vspace{-0.4cm}
\begin{align} \label{ser1}
 f^{(c)}_{r_{1}}( r ) &= 2\lambda \pi r e^{-\lambda\pi r^2}; \quad 0 \leq r \leq \infty
\end{align}

The marginal PDF of the distance between the test user and its serving BS in the blackout case without BS cooperation is expressed as:
\begin{align}
\vspace{-0.1cm}
\hspace{0.3cm}
f^{(bk)}_{r_{2}}(y) =2(\lambda\pi)^2y^3e^{-\lambda\pi y^2};\quad 0 \leq y \leq\infty
\label{ser_dist}
\end{align}
where $y$ represents the distance between the test user and second nearest BS, which is the serving BS in blackout case in non-CoMP mode.
The joint PDF of the distances between the test user and its serving/cooperating BSs in the blackout case with BS cooperation is given by:
\begin{align}
\hspace{0.3cm}
f^{(bk)}_{r_2,r_3}(y,z)&=& 4(\pi\lambda)^3 y^{3}z e^{-\pi\lambda z^{2}};\quad 0 \leq y \leq z \leq \infty
\label{joint_m}
\end{align}

The conditional PDF of $r_1$ with respect to (w.r.t.) $r_2$, for the blackout case, can be expressed as:
\begin{align}
\hspace{0.3cm}
f^{(bk)}_{r_{1}}( x|r_2 )=\frac{2x}{r_{2}^{2}};\quad 0 \leq x \leq r_2 \leq \infty
\label{cond}
\end{align}
\end{col}
\begin{proof}
\textit{The marginal PDF of $r_1$ in \eqref{ser1} is obtained by integrating \eqref{eq:joint} w.r.t. $y$ and $z$ where $y$ and $z$ are bounded as $x \leq y \leq \infty$ and $y \leq z \leq \infty$, respectively. The marginal PDF of $r_2$ in \eqref{ser_dist} is obtained by integrating \eqref{eq:joint} w.r.t. $x$ and $z$, where $x$ and $z$ are bounded as $0 \leq x \leq y$ and $y \leq z \leq \infty$, respectively. Similarly, the joint PDF of $r_2$ and $r_3$ is obtained by integrating \eqref{eq:joint} w.r.t. $x$ from 0 to $y$. The conditional PDF of $r_1$ in \eqref{cond} is derived by dividing the joint PDF in \eqref{eq:joint} by the joint marginal distribution obtained in \eqref{joint_m}.}
\end{proof}
\section{Coverage Probability}
The coverage probability is defined as the probability that the received SINR exceeds a specified threshold $T$. For the best connected case, the coverage probability can be expressed as:
\begin{eqnarray*}
\hspace{0.4cm} \mathcal{C}_c &=& \mathbb{P} \left\{ \frac{P h_1 \parallel r_{1}\parallel^{-\eta}}{ \sum_{i\epsilon \Phi \backslash b_1}{}P h_{i} \parallel r_{i}\parallel^{-\eta} + \sigma^2} > T \right\}
\end{eqnarray*}
where $b_1$ denotes the serving BS. Therefore, its power is excluded from the aggregate interference term in the denominator. The coverage probability for the blackout case with and without BS cooperation is given by:
\begin{align*}
\hspace{0.4cm} \mathcal{C}_{bk}^{(1)}=\mathbb{P}\left\{\frac{P h_2\parallel r_{2}\parallel^{-\eta}}{Ph_{1} \parallel r_{1}\parallel^{-\eta}+I_{r_1}+\sigma^2} > T  \right\}
\end{align*}
\begin{align*}
\mathcal{C}_{bk}^{(2)}=\mathbb{P} \left\{\frac{|\sqrt{P_{2}}h_{2}\parallel r_{2}\parallel^{-\eta/2}+\sqrt{P_{3}}h_{3}\parallel r_{3}\parallel^{-\eta/2}|^2}{P_{1}h_{1} \parallel r_{1}\parallel^{-\eta}+I_{r_2}+\sigma^2}>T\right\}
\end{align*}
where $C_{bk}^{(1)}$ and $C_{bk}^{(2)}$ are the coverage probabilities in the blackout case without and with BS cooperation, respectively. Also, $I_{r_1}$ and $I_{r_2}$ denote the aggregate inference powers in blackout without and with BS cooperation, respectively. Thus, we define $I_{r_1}$ and $I_{r_2}$ as follows:
\vspace{-0.2cm}
\begin{align*}
I_{r_1}= \sum_{i\epsilon\Phi\backslash b_{1},b_{2}}^{} P_{i}h_{i}\parallel r_{i}\parallel^{-\eta},\hspace{0.4cm} I_{r_2}= \sum_{i\epsilon\Phi\backslash b_{1},b_{2},b_{3}}^{} P_{i}h_{i}\parallel r_{i}\parallel^{-\eta}
\end{align*}
From \cite{icc} and \cite{comp4}, given that $h_{i}$s are $iid$ $\mathcal{CN}$(0,1) such that $|x_{2}h_2+x_{3}h_3|^2\sim \exp(\frac{1}{x_{2}^{2}+x_{3}^{2}})$, we can write the conditional coverage probability for the best connected user as:
\begin{eqnarray} \label{cov1}
\mathcal{C}_c(r_1) = e^{- \frac{ T \sigma^2 r_1^\eta}{P}}\mathscr{L}_{I_r}\left(\frac{T r_1^\eta}{P}\right)
\end{eqnarray}
Similarly, we can write the conditional coverage probability (conditioning on $r_2$) for a blackout user without BS cooperation as:
\begin{eqnarray}  \label{cov2}
\mathcal{C}_{bk}^{(1)}(r_2)= e^{- \frac{ T \sigma^2 r_2^\eta}{P}} \mathscr{L}_{I_1}\left(\frac{T r_2^\eta}{P}\right) \mathscr{L}_{I_{r_1}}\left(\frac{T r_2^\eta}{P}\right)
\end{eqnarray}
The conditional coverage probability (conditioned on $r_2$ and $r_3$) for a blackout user with cooperative service is given by:
\begin{align}
\hspace{-0.2cm}C_{bk}^{(2)}(r_2,r_3)\hspace{-0.1cm}=\exp\bigg(\hspace{-0.05cm}\frac{-T\sigma^2}{x_{2}^{2}\hspace{-0.1cm}+\hspace{-0.1cm}x_{3}^{2}}\hspace{-0.05cm}\bigg)\mathscr{L}_{I_1}\hspace{-0.1cm}\bigg(\hspace{-0.05cm}\frac{T}{x_{2}^{2}\hspace{-0.1cm}+\hspace{-0.1cm}x_{3}^{2}}\bigg)\mathscr{L}_{I_{r_2}}\bigg(\hspace{-0.05cm}\frac{T}{x_{2}^{2}\hspace{-0.1cm}+\hspace{-0.1cm}x_{3}^{2}}\hspace{-0.05cm}\bigg)
\label{concv}
\end{align}
where
\vspace{-0.1cm}
\begin{eqnarray*}
x_{i}=\sqrt{P_{i}}\parallel r_{i}\parallel^{-\eta/2}, \quad I_{1}= P_{1}h_{1} \parallel r_{1}\parallel^{-\eta}
\end{eqnarray*}
The Laplace transforms (LTs) of $I_1$ and $I_{r_1}$ for the best connected and blackout cases without BS cooperation are characterized in \cite{icc}. Therefore, we focus on the characterization of the LTs of $I_1$ and $I_{r_2}$ in the blackout case with BS cooperation. Here, it is worth mentioning that these LTs are different from \cite{comp3} due to different cooperating BSs. We consider the cooperation between second and the third nearest BSs. The LTs of $I_1$ and $I_{r_2}$ in blackout with BS cooperation are expressed in following lemma.
\begin{lem}
The Laplace transform of $I_{1}$ in the blackout case with BS cooperation is given by
\begin{eqnarray}
\mathscr{L}_{I_1}(s)&=&\int_{0}^{r_2}\frac{2r_{1}}{r_{2}^{2}(1+sPr_{1}^{-\eta})}d_{r_1}
\end{eqnarray}
The Laplace transform of $I_{r_2}$ can be expressed in terms of a hypergeometric function as:
\small
\begin{align}
\hspace{-0.3cm}\mathscr{L}_{I_{r_2}}\hspace{-0.1cm}(s)\hspace{-0.1cm}&=&\hspace{-0.35cm}\exp\bigg(\hspace{-0.1cm}\frac{-2\pi\lambda sPr_{3}^{2-\eta}}{\eta-2}\mathstrut_2 F_1\big(1,1\hspace{-0.1cm}-\hspace{-0.1cm}\frac{2}{n};2\hspace{-0.1cm}-\hspace{-0.1cm}\frac{2}{n};-sP r_{3}^{-\eta}\big)\hspace{-0.1cm}\bigg)
\end{align}
\normalsize
where
\begin{equation*}
s=\frac{T}{P(r_{2}^{-\eta}+r_{3}^{-\eta})}
\end{equation*}
\label{Lts}
\end{lem}
\begin{proof}
\textit{$\mathscr{L}_{I_1}(s)$ and $\mathscr{L}_{I_{r_2}}(s)$ are obtained by following the same procedure in \cite{icc} while considering $s=\frac{T}{P(r_{2}^{-\eta}+r_{3}^{-\eta})}$. Also, we perform the integration from $r_3$ to $\infty$ while calculating the Laplace transform for $I_{r_2}$.}
\end{proof}
We evaluate the above LTs for a special case at $\eta=4$, which is the most practical value in outdoor environment. The LTs at $\eta=4$ are given by the following corollary.
\begin{col}
The Laplace transforms of $I_{1}$ and $I_{r_{2}}$ at $\eta=4$ for the blackout case with BS cooperation are boiled down into much simpler expressions as shown below:
\begin{eqnarray}
\mathscr{L}_{I_1}\left(s\right)\big|_{\eta=4}=1-\sqrt{\frac{Tr_{3}^{4}}{r_{2}^{4}+r_{3}^{4}}}\arctan\bigg(\sqrt{\frac{r_{2}^{4}+r_{3}^{4}}{Tr_{3}^{4}}}\bigg)
\end{eqnarray}
\vspace{-0.1cm}
\begin{align}
\hspace{-0.37cm}\mathscr{L}_{I_{r_2}}\hspace{-0.1cm}\left(s\right)\big|_{\eta=4}\hspace{-0.05cm}=\exp\hspace{-0.05cm}\bigg(\hspace{-0.15cm}-\hspace{-0.1cm}\pi\lambda\sqrt{\hspace{-0.1cm}\frac{T}{r_{2}^{-4}\hspace{-0.1cm}+\hspace{-0.1cm}r_{3}^{-4}}}\arctan\hspace{-0.1cm}\bigg(\hspace{-0.05cm}\sqrt{\frac{Tr_{2}^{4}}{r_{2}^{4}\hspace{-0.1cm}+\hspace{-0.1cm}r_{3}^{4}}}\bigg)\hspace{-0.05cm}\bigg)
\end{align}
\end{col}
The coverage probabilities for the best connected and blackout cases with and without BS cooperation are illustrated in the following theorem.
\begin{theorem}
\label{theorem:Coverage Probability}
\begin{figure*}
\begin{eqnarray}
\mathcal{C}_c = 2 \pi \lambda x\int_0^\infty\exp\left(- \frac{ T \sigma^2 x^\eta}{P} - \pi\lambda x^{2} \left(1+ T^{2/\eta}\int_{T^{-2/\eta}}^{\infty}\frac{1}{1+w^{\eta/2}}dw\right) \right) dx
\label{f_cov1}
\\ \notag
\end{eqnarray}
\begin{eqnarray}
\mathcal{C}_{bk}^{(1)} = (2\lambda\pi)^{2}\int_{0}^{\infty}y \exp\bigg({-Ty^{\eta}\sigma^{2}-\lambda\pi y^2 \bigg(\frac{2 T}{\eta-2}\mathstrut_2 F_1\big(1,1-\frac{2}{\eta};2-\frac{2}{\eta};-T\big)+1\bigg)}\bigg)\int_{0}^{y}\frac{x}{1+Ty^{\eta}x^{-\eta}}dxdy
\label{f_cov2}
\\ \notag
\end{eqnarray}
\begin{eqnarray}
\mathcal C_{bk}^{(2)}=8(\pi\lambda)^3\int_{0}^{\infty}\hspace{-0.1cm}\int_{r_2}^{\infty}r_{2}r_{3} \exp\bigg(\hspace{-0.1cm}-\hspace{-0.1cm}s\sigma^2\hspace{-0.1cm}-\hspace{-0.1cm}\pi\lambda r_{3}^{2}\bigg(1+\frac{2sPr_{3}^{-\eta}}{\eta-2}\mathstrut_2 F_1\big(1,1\hspace{-0.1cm}-\hspace{-0.1cm}\frac{2}{n};2\hspace{-0.1cm}-\hspace{-0.1cm}\frac{2}{n};-sP r_{3}^{-\eta}\big)\hspace{-0.1cm}\bigg)\bigg)\int_{0}^{r_2}\frac{r_{1}}{1+sPr_{1}^{-\eta}}dr_{1} dr_{3}dr_{2}
\label{f_covbk}
\end{eqnarray}
\vspace{-0.5cm}
\hrulefill
\end{figure*}
\begin{figure}[!t]
\centering
\includegraphics[width=0.935 \linewidth]{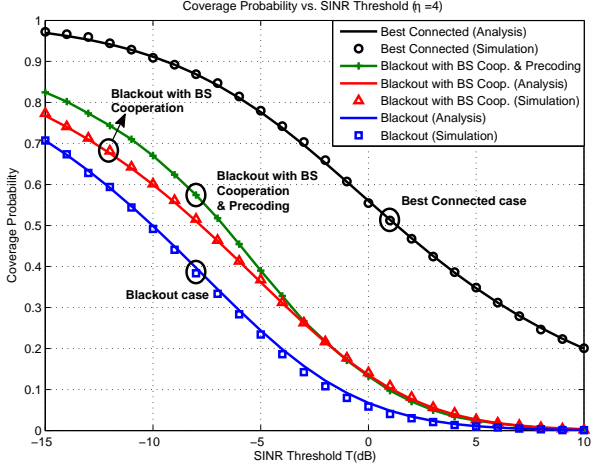}
\small \caption{Coverage probability plots for best connected and HO skipping cases evaluated at $\eta=4$.}
\label{cp}
\end{figure}
Considering a PPP single tier downlink cellular network with BS intensity $\lambda$ in a Rayleigh fading environment, the coverage probability for the best connected, blackout and blackout with BS cooperation users can be expressed by \eqref{f_cov1}, \eqref{f_cov2} and \eqref{f_covbk}, respectively. In an interference limited environment with a special case at $\eta=4$, \eqref{f_covbk} reduces to
\small
\begin{align*}
\hspace{-0.2cm}C_{bk}^{(2)}\hspace{-0.1cm}=\hspace{-0.2cm}\int_{0}^{\infty}\hspace{-0.1cm}\int_{r_2}^{\infty}\hspace{-0.2cm}4(\pi\lambda)^3r_{2}^{3}r_{3}\bigg(\hspace{-0.1cm}1-\hspace{-0.1cm}\sqrt{\frac{Tr_{3}^{4}}{r_{2}^{4}+r_{3}^{4}}}\arctan\bigg(\hspace{-0.1cm}\sqrt{\frac{r_{2}^{4}+r_{3}^{4}}{Tr_{3}^{4}}}\bigg)\bigg)\notag\cdot\\
\end{align*}
\vspace{-1cm}
\begin{align}
\hspace{-0.4cm}\exp\bigg(\hspace{-0.2cm}-\hspace{-0.1cm}\pi\lambda\bigg(\hspace{-0.1cm}r_{3}^{2}\hspace{-0.1cm}+\hspace{-0.1cm}\sqrt{\frac{T}{\hspace{-0.1cm}r_{2}^{-4}\hspace{-0.1cm}+\hspace{-0.1cm}r_{3}^{-4}}}\arctan\bigg(\hspace{-0.1cm}\sqrt{\frac{Tr_{2}^{4}}{r_{2}^{4}+r_{3}^{4}}}\bigg)\bigg)\bigg)dr_{3}dr_{2}
\end{align}
\normalsize
\end{theorem}
\begin{proof}
\textit{We prove the theorem by substituting the LTs of $I_{1}$, $I_{r_1}$ and $I_{r_2}$ from \cite{icc} and Lemma 2 in the conditional coverage probabilities obtained in \eqref{cov1}, \eqref{cov2} and \eqref{concv} and integrating it over the distance distributions obtained in Corollary 1.}
\end{proof}
Fig.~\ref{cp} shows the coverage probabilities of the best connected vs. blackout cases with and without BS cooperation. In the blackout case, the user is connected to the second nearest BS while receiving huge interference from the nearest BS. In such a case, some interference cancellation technique can be employed to further enhance the SINR, which will lead to an improved coverage probability. In the literature, there is a technique in which a particular interfering signal can be serially detected, demodulated, decoded and removed from the aggregate interference. Thus, by following \cite{16a}, we evaluate our model with nearest BS interference cancellation (IC).
\begin{theorem}
In a PPP downlink cellular network with BS cooperation and IC capabilities, coverage probability in the blackout case is given by:
\vspace{-0.18cm}
\begin{eqnarray*}
C_{bk,IC}^{(2)}=\int_{0}^{\infty}\int_{r_2}^{\infty}4(\pi\lambda)^3 r_{2}^{3} r_{3}\exp(-s\sigma^{2}-\pi\lambda r_{3}^{2})\notag\cdot
\end{eqnarray*}
\vspace{-0.35cm}
\begin{align}
\hspace{-0.3cm}
\exp\bigg(\hspace{-0.1cm}\frac{-2\pi\lambda sPr_{3}^{2-\eta}}{\eta-2}\mathstrut_2 F_1\big(1,1\hspace{-0.1cm}-\hspace{-0.1cm}\frac{2}{n};2\hspace{-0.1cm}-\hspace{-0.1cm}\frac{2}{n};-sP r_{3}^{-\eta}\big)\hspace{-0.1cm}\bigg)dr_{3}dr_{2}
\end{align}
\end{theorem}
\begin{proof}
\textit{We derive this relation by following the same procedure as for Theorem 1 and eliminating the LT of $I_1$.}
\end{proof}
Fig.~\ref{cp2} shows the coverage probabilities for the best connected, blackout and blackout with BS cooperation users when nearest BS IC is enabled. It can be observed that the coverage probability for the blackout case with BS cooperation and IC follows the best connected trend at low threshold values.

In Fig.~\ref{cp} and~\ref{cp2}, we show by simulation that the CSI aware BS cooperation (i.e. transmission with precoding) offers marginal coverage probability gains when compared to the non-coherent BS cooperation. Particularly, the precoding offers $6\%$ and $8\%$ increase in coverage probability w.r.t. the non-coherent BS cooperation without and with IC, respectively. Furthermore, the gains diminish and approach zero at high SINR thresholds. It is worth mentioning that CSI aware BS cooperation coverage expressions can be derived by following~\cite[Theorem 5]{comp3} but with the joint service distance distribution obtained in \eqref{joint_m}.
\begin{figure}[!t]
\centering
\includegraphics[width=0.935 \linewidth]{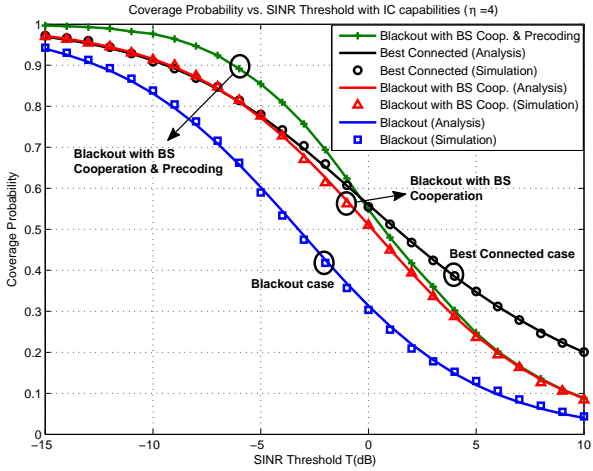}
\small \caption{Coverage probability plots for best connected and HO skipping cases with IC evaluated at $\eta=4$.}
\label{cp2}
\end{figure}
\section{Average Throughput}
In this section, we encompass the user mobility in the rate analysis and derive an expression for the average throughput. We focus on the user throughput after excluding signaling and control overheads. Therefore, we need to omit the control overhead from the overall capacity. According to 3GPP Release 11 \cite{17a}, the control overhead consumes a fraction $u_c=0.3$ of the overall network capacity. Thus, the average throughput of a typical user can be expressed as:
\vspace{-0.1cm}
 \begin{eqnarray}
\text{AT} &=& W\mathcal{R}(1-u_c)(1-D_{HO})
\end{eqnarray}
 where W is the overall bandwidth, $\mathcal{R}$ is the spectral efficiency and $D_{HO}$ is the HO cost. We define $D_{HO}$ as the time wasted in performing HO per unit time, which is a unit-less quantity that defines the percentage of time wasted in HO signalling. $D_{HO}$ is a function of the number of HOs per second (HO rate) and the time required for performing a single HO. $D_{HO}$ is given by:
\vspace{-0.1cm}
\begin{eqnarray}
D_{HO}&=& H(v) * d
\end{eqnarray}
The HO rate can be defined as the number of intersections between the user's trajectory and the cell boundaries per unit time. Therefore, it is a function of user velocity and BS intensity. Following \cite{10a}, we can define the HO rate for a typical user in a single tier network as:
\vspace{-0.1cm}
\begin{eqnarray}
H(v)&=& \frac{4v}{\pi}\sqrt{\lambda}
\end{eqnarray}
From \cite{6a}, the average spectral efficiency can be written as:
\begin{eqnarray} \label{rates}
\mathcal{R}&\stackrel{(a)}{=}& \int_{0}^{\infty}\mathbb{P}\left\{\ln(1+ {\rm SINR})>z\right\}dz\\
&\stackrel{(b)}{=}&\int_{0}^{\infty}\frac{\mathbb{P}\left\{{\rm SINR}>t\right\}}{t+1}dt
\end{eqnarray}
where (a) follows because $\ln(1+{\rm SINR})$ is a positive random variable and (b) follows by the change of variables $t=e^{z}-1$ \cite{18a}. By performing the numerical evaluation for spectral efficiency in the conventional, blackout with and without cooperation case, we get the spectral efficiency in nats/sec/Hz as shown in table \ref{tab1}.

\begin{table}[ht]
\renewcommand{\arraystretch}{1.3}
\caption{Spectral Efficiency for all cases in nats/sec/Hz}
\label{tab1}
\centering
 \begin{tabular}{||c c c||}
 \hline
 \multicolumn{3}{||c||}{\textbf{Spectral Efficiency (nats/sec/Hz)}} \\
 \hline
 Case & Non-IC & IC \\
 \hline\hline
 Best connected $\mathcal{R}_{c}$ & 1.49 & - \\
 \hline
 Blackout $\mathcal{R}_{bk}^{(1)}$ & 0.21 & 0.66\\
 \hline
 Blackout (BS coop.) $\mathcal{R}_{bk}^{(2)}$ & 0.31 & 1.01\\
 \hline
\end{tabular}
\end{table}

In the HO skipping case, the user alternates between the best connected and blackout cases. Therefore, we can assume that the user spends 50$\%$ time in the best connected mode and 50$\%$ in the blackout mode with/without BS cooperation. The average spectral efficiency for the HO skipping case is given by:
\vspace{-0.2cm}
\begin{eqnarray}
\mathcal{R}_{sk}^{(1)}= \frac{\mathcal{R}_{c}+\mathcal{R}_{bk}^{(1)}}{2}&\simeq&0.85\hspace{0.2cm}\text{nats/sec/Hz}\\
\mathcal{R}_{sk}^{(2)}= \frac{\mathcal{R}_{c}+\mathcal{R}_{bk}^{(2)}}{2}&\simeq&0.90\hspace{0.2cm}\text{nats/sec/Hz}
\end{eqnarray}
where $\mathcal{R}_{sk}^{(1)}$ and $\mathcal{R}_{sk}^{(2)}$ represent the average spectral efficiency in the HO skipping case without and with BS cooperation, respectively. Similarly, the average spectral efficiency for the IC enabled HO skipping case without and with BS cooperation are found to be 1.08 nats/sec/Hz and 1.25 nats/sec/Hz, respectively.
\section{Numerical Results}
In this section, we compare the throughput performance of our developed model with the best connected and blackout cases. The simulation parameters are shown in table \ref{tab2}. We conduct our analysis based on the nearest BS interference cancellation and consider various values for the HO delay and BS intensity.
\begin{table}[ht]
\renewcommand{\arraystretch}{1.3}
\caption{Simulation parameters}
\label{tab2}
 \begin{tabular}{||l|l||}
 \hline
 \multicolumn{2}{||c||}{\textbf{Simulation parameters}} \\
 \hline
\hline  Tx Power: \quad 1 watt &Path loss exponent $\eta$:\quad 4\\
\hline  Overall Bandwidth $W$ : \quad 10 MHz &HO delay $d$:\quad 0.7, 2 s\\
\hline  Control Overhead $u_{c_{conv}}$: \quad 0.3 & Control Overhead $u_{c_{bk}}$: \quad 0.15\\
 \hline
          \end{tabular}
        \end{table}
 \begin{figure}[]
\centering
\includegraphics[width=0.935 \linewidth]{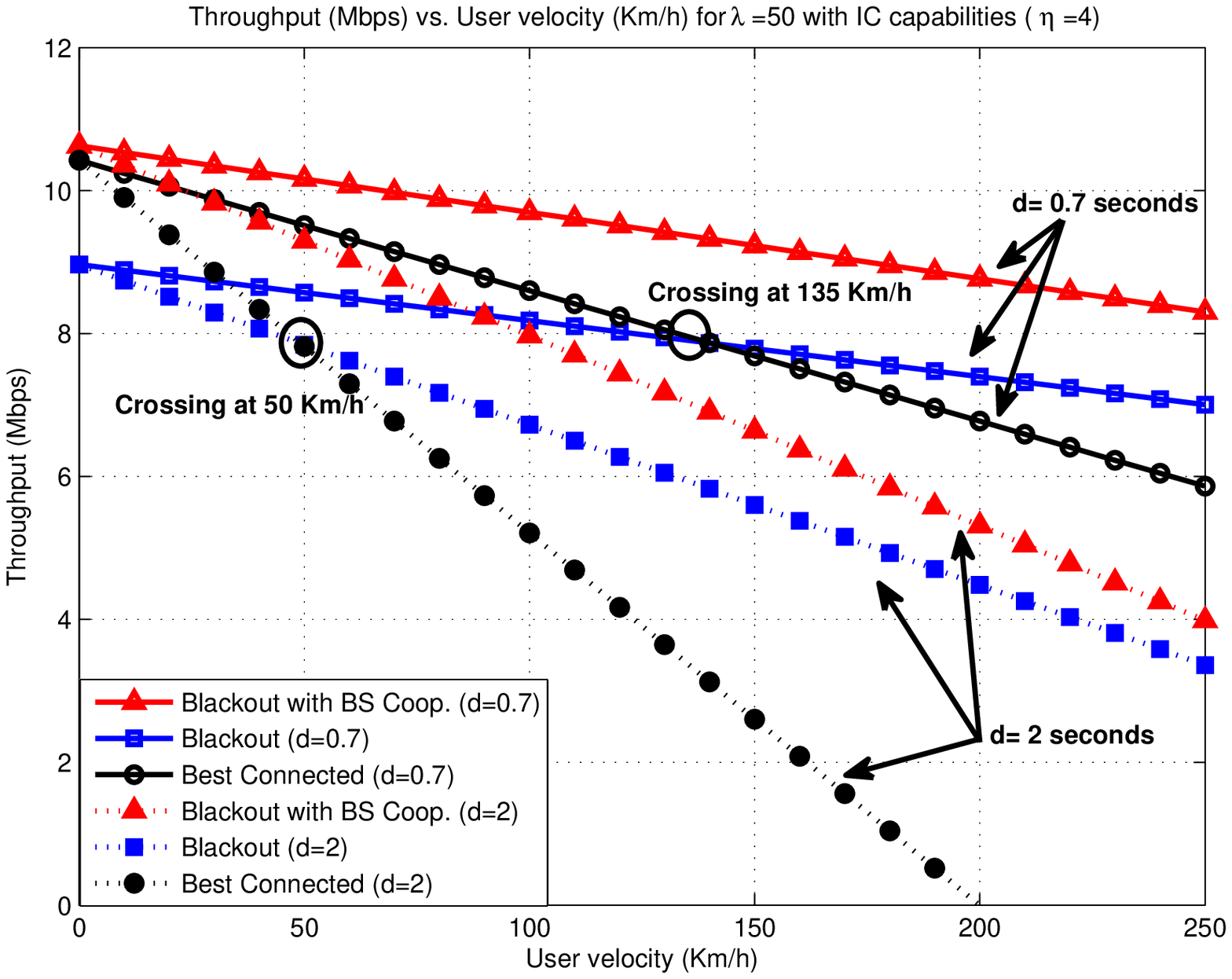}
\small \caption{Average Throughput (Mbps) vs. user velocity (Km/h) for $\lambda=50 BS/Km^2$}
\label{T1}
\end{figure}
\begin{figure}[!h]
\centering
\includegraphics[width=0.935 \linewidth]{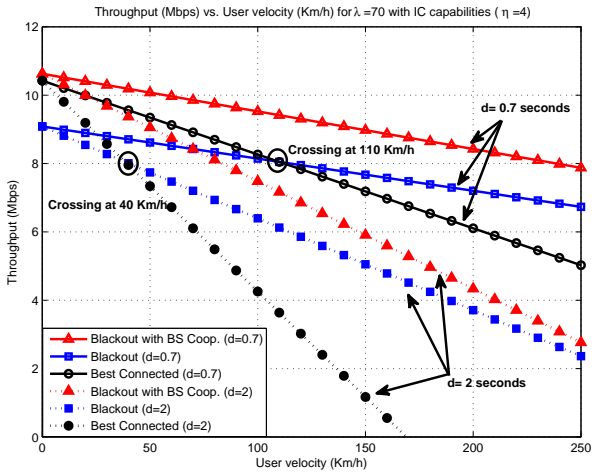}
\small \caption{Average Throughput (Mbps) vs. user velocity (Km/h) for $\lambda=70 BS/Km^2$}
\label{T2}
\end{figure}
Fig.~\ref{T1} and~\ref{T2} show performance gains in the average throughput of a blackout over conventional user. It is observed that the blackout user with BS cooperation shows considerable gains in the average throughput as compared to the conventional and blackout users without BS cooperation. For instance, the cooperative blackout user moving at the speed of 100 Km/h, in a cellular network with the BS intensity 70 $BS/Km^{2}$ and single HO delay $d=0.7 s$, will experience performance gains of 15$\%$ and 17$\%$ as compared to the conventional and blackout cases, respectively.
\section{Conclusion}
This paper proposes a cooperative HO skipping scheme for a single tier cellular network to enhance the average throughput for mobile users. We develop an analytical paradigm to model the performance of the proposed cooperative HO skipping scheme in order to study the effect of HO delay on the average throughput. The developed mathematical model is based on stochastic geometry and is validated via Monte Carlo simulations. The results manifest the negative impact of HO on the users' throughput in dense cellular networks and emphasize the potential of the proposed HO scheme to mitigate such negative HO impact. Particularly, the results show up to $56\%$ more rate gains, which can be harvested via the proposed cooperative HO scheme when compared to the conventional HO scheme that always maintains the best BS association.
For future work, we will extend our study towards a multi-tier network with user velocity aware HO skipping. Thus, we will propose multiple types of HO skipping procedures based on different user mobility profiles.
\bibliographystyle{IEEEtran}
\bibliography{Rabe.bib}
\vfill

\end{document}